\documentclass[a4paper, 10pt]{llncs}
\usepackage{bbm}
\usepackage{amsmath}\allowdisplaybreaks[4]
\usepackage{amsfonts}
\usepackage{pifont}
\usepackage{amssymb}
\usepackage[sl]{caption}
\usepackage{psfig}
\usepackage{fancyhdr}

\usepackage{graphicx}
\usepackage{color}
\usepackage{listings}
\usepackage{paralist}
\usepackage[boxed,lined]{algorithm2e}
\usepackage[T1]{fontenc}

\providecommand{\Appendix}{}
\renewcommand{\Appendix}[2][?]{%
        \refstepcounter{section}%
        \vspace{\parskip}%
        {\flushright\large\bfseries\appendixname\ \thesection: #1}%
        \vspace{\baselineskip}%
}

\renewcommand{\appendix}{%
        \newpage
        \renewcommand{\section}{\secdef\Appendix\Appendix}%
        \renewcommand{\thesection}{\Alph{section}}%
        \setcounter{section}{0}%
}

\newcommand{\initOneLiners}{%
    \setlength{\itemsep}{0pt}
    \setlength{\parsep }{0pt}
    \setlength{\topsep }{0pt}
}

\newcommand{\ignore}[1]{}

\newcommand{\shortv}[1]{}

\begin{document}

\title{($1+\epsilon$)-Distance Oracle on Planar, Labeled Graph}
\author{Mingfei Li, Christoffer Ma, and Li Ning \\
\small{Department of Computer Science, The University of Hong Kong}}

\maketitle

\begin{abstract}
Given a vertex-labeled graph, each vertex $v$ is attached with a label from
a set of labels. The vertex-label query desires the length of the shortest
path from the given vertex to the set of vertices with the given label. We show how to construct
an oracle if the given graph is planar, such that $O(\frac{1}{\epsilon}n\log n)$ storing space is needed,
and any vertex-label query could be answered in $O(\frac{1}{\epsilon}\log n\log \rho)$ time
with stretch $1+\epsilon$. $\rho$ is the radius of the given graph, which is half of the diameter.
For the case that $\rho = O(\log n)$, we construct an oracle that achieves $O(\log n)$ query time,
without changing the order of storing space.
\end{abstract}

\section{Introduction}

We consider those undirected graphs, in which each vertex is attached with a label from
a set of labels, denoted by $L$. Fixed such a graph $G=(V,E)$ and the label set $L$, the distance
between two nodes $v,u\in V$, denoted by $\delta(v,u)$ is the length of the shortest path connecting $v$
and $u$ in $G$, and the distance between a vertex $u\in V$ and a label $\lambda\in L$, denoted by $\delta(u,\lambda)$,
is the distance between $u$ and node $v$ that is closest to $u$ among all nodes with label $\lambda$,
i.e. $\delta(u,\lambda) = \min\{\delta(u,v)| v \textrm{ is attached with label $\lambda$}\}$.
In the applications involving graphs,
the query of vertex-label distance is often asked and used as a basic sub-procedure to achieve more complicated
task. For example, a navigation software need to answer how far is the closest store for a specified
service from the current position. Since these kind of questions raise very frequently, the answer
should be returned with as less time as possible. Trivially, people could precalculate and store
the answer for all possible queries. However, this may take too much space, which is of order $O(|V|\times |L|)$.

The aim of \emph{distance oracle} is to precalculate and store information using less than $O(|V|\times |L|)$
space, such that any distance query could be answered more efficiently than process the calculation
based only on the graph structure. The approximate distance oracle answers the query with some stretch.
In details, a distance oracle with stretch $1+\epsilon$ returns $d(u,\lambda)$ as the approximation
to $\delta(u,\lambda)$, such that $\delta(u,\lambda) \leq d(u,\lambda) \leq (1+\epsilon)\delta(u,\lambda)$.

Since in many practical cases, the given graph is drawn on a plane, it has
wide applications to derive the particular distance oracle for planar graphs.

\subsection{Related Work}

\noindent\paragraph{Vertex-Label Distance Oracle.}
The problem of construct approximate distance oracles for vertex-labeled graphs was formalized and studied
by Hermelin et al. in \cite{Hermelin2011}. Let $n$ denote the number of nodes and $m$ denote the number of edges.
They adapted the approximate scheme introduced by Thorup and Zwick
in \cite{Thorup2001} to show the construction of vertex-label distance oracles
with expected size $O(kn^{1+\frac{1}{k}})$, stretch $4k-5$ and query time $O(k)$.
The preprocessing time is $O(kmn^{\frac{1}{k}})$. Let $l = |L|$.
They also constructed vertex-label distance oracles with expected size  $O(knl^{\frac{1}{k}})$,
stretch $2^k-1$ and query time $O(k)$. The preprocessing time is $O(kmn^{\frac{k}{2k-1}})$.
For a vertex in the graph, the associated label may change. A simple way to support label changes is to
construct a new distance oracle. In \cite{Hermelin2011}, they constructed vertex-label distance oracles
with expected size $O(kn^{1+\frac{1}{k}})$, stretch $2\cdot 3^{k-1} + 1$ and query time $O(k)$,
which can support label changes in $O(kn^{\frac{1}{k}}\log n)$ time. In \cite{Chechik2011},
Chechik showed that Thorup and Zwick's scheme
could also be modified to support label changes in $O(n^{\frac{1}{k}}\log^{1-\frac{1}{k}}n \log\log n)$ time,
with the expected size $\tilde{O}(n^{1+\frac{1}{k}})$, stretch $4k-5$ and query time $O(k)$. The
preprocessing time is $O(kmn^{\frac{1}{k}})$.

The vertex-label distance oracle has also been studied for some specified class of graphs. Tao et al. have shown
how to construct vertex-label distance oracles for XML trees, in \cite{Tao2011}. For the case that
each node is assigned with exactly one label, their construction results in exact vetex-label
distance oracles with size $O(n)$, and query time $O(\log n)$. The preprocessing time is $O(n\log n)$.

\noindent\paragraph{Vertex-Vertex Distance Oracle.}
In \cite{Thorup2001}, Thorup and Zwick have introduced a well-known scheme to
construct vertex-vertex distance oracle with expected size $O(kn^{1+\frac{1}{k}})$, stretch $2k-1$,
and query time $O(k)$. The preprocessing time is $O(kmn^{\frac{1}{k}})$.
Wulff-Nilsen in \cite{Wulff-Nilsen2012} has improved the preprocessing time to
$O(\sqrt{k}m+kn^{1+\frac{c}{\sqrt{k}}})$ for some universal constant $c$, which
is better than $O(kmn^{\frac{1}{k}})$ except for very sparse graphs and small $k$.
For planar graphs, Klein in \cite{Klein2002} has shown how to construct
vertex-vertex distance oracles with size $O(\frac{1}{\epsilon}n\log n)$, stretch $1+\epsilon$ and
query time $O(1)$.

\noindent\paragraph{Shortest Path.}
The construction of distance oracles often harness the shortest path algorithms in preprocessing stage.
Although it is better to know as well as possible the methods that aim at calculating the shortest path,
we only selected the most related ones and list them here. For the others, we will introduce them while
they are used in our algorithm.

A \emph{shortest path tree} with vertex $v$ is a
tree rooted at $v$ and consisting of all nodes and a subset of edges from the given graph, such
that for any $u$ in the given graph, the path from $v$ to $u$ in the tree is the shortest path from
$v$ to $u$ in the original graph. Given a single vertex, to calculate the shortest path tree rooted at it
is called \emph{single source shortest path} problem.
In undirected graphs, the single source shortest path tree could be calculated in time of order $O(m)$ where
$m$ is the number of the edges in the given graph. The algorithm is introduced by Throup in \cite{Thorup1997}.
In directed graphs, the single source shortest path tree could be calculated in time of order $O(m + n\log n)$ where
$n$ is the number of the nodes in the given graph. This is done by the well known Dijkstra algorithm using
Fibonacci heap \cite{Cormen2001}.

\subsection{Simple Solution in Doubling Metrics Spaces.}
If the metric implied by the given graph is doubling, the following procedure provides a
simple solution to return $\delta(u,\lambda)$ with $(1+\epsilon)$-stretch.

\noindent\paragraph{Preprocessing.} Let $\epsilon' = \frac{\epsilon}{3}$. For $\epsilon < 1$,
$(1+\epsilon')^2 < 1+\epsilon$.
For each label $\lambda \in L$, construct the oracle to support
$(1+\epsilon')$-nearest neighbor search. In additional, construct the oracle to support
$(1+\epsilon')$ vertex-vertex distance query.

\noindent\paragraph{Query.} Given $u\in V$ and $\lambda\in L$, find the $(1+\epsilon')$-NN of $u$
among the nodes with label $\lambda$, and then query for their $(1+\epsilon')$ distance.

\noindent\paragraph{Space and Query Time.} The oracle supporting approximate nearest neighbor search
for $\lambda$ could be constructed using $O(n_{\lambda})$ space \ignore{\cite{}},
where $n_i$ is the number of nodes with label
$\lambda$. Since any node is allowed to attached with only one label, then the space used in all is $O(n)$.
This kind of oracle could answer the query in $O(\log n_{\lambda}) = O(\log n)$ time \ignore{\cite{}}.
The oracle supporting approximate vertex-vertex query distance could be constructed using $O(n)$ space and
answer the query in $O(1)$ time \ignore{\cite{}}.

\subsection{Our Contribution}
As shown in the subsection of related work, we are not aware of any vertex-label distance oracle on planar graphs.
In this paper, we mainly show the following two results.

\begin{theorem}
Given an undirected planar graph $G=(V,E)$ and a label set $L$, each vertex $v\in V$ is attached with
one label in $L$. For any $0< \epsilon < 1$, there exists an oracle that could answer
any vertex-label query with stretch $1+\epsilon$, in $O(\frac{1}{\epsilon}\log n\log \rho)$ time.
The oracle needs $O(\frac{1}{\epsilon}n\log n)$ space.
\end{theorem}

\begin{theorem}
Given an undirected planar graph $G=(V,E)$ and a label set $L$, each vertex $v\in V$ is attached with
one label in $L$. If the radius of $G$ is of order $O(\log n)$, then
for any $0< \epsilon < 1$, there exists an oracle that could answer
any vertex-label query with stretch $1+\epsilon$, in $O(\frac{1}{\epsilon}\log n)$ time.
The oracle needs $O(\frac{1}{\epsilon}n\log n)$ space.
\end{theorem}

\setcounter{theorem}{0}

\section{Preliminary}
\noindent\paragraph{Lipton Tarjan Separator. \cite{Lipton1979}} Let $T$ be a spanning tree of a planar embedded
triangulated graph $G$ with weights on nodes. Then there is an edge $e\not\in T$, s.t. the strict interior
and strict exterior of the simple cycle in $T\cup \{e\}$ each contains weight no more than $\frac{2}{3}$
of the total weight.

\noindent\paragraph{Recursive Graph Decomposition \cite{Klein2002}.} The recursive graph decomposition (RGD) of a given graph
$G$ is a rooted tree, such that each vertex $p$ in $G$ maintains
\begin{itemize}
\item
a set $N(p)$ of nodes in $G$, in particular the root of RGD maintains (as a label) $N(p) = V(G)$, and
\item
$p$ is a leaf of RGD \emph{iff.} $N(p)$ contains only one node of $G$, in this case let $S(p) = N(p)$;
\item
if $p$ is not a leaf of RGD, it maintains (as a label) an $\alpha$-balanced separator $S(p)$ of $G$, balanced
with respect to the weight assignment in which each node in $N(p)$ is assigned weight $1$ and other
nodes are assigned weight $0$.
\end{itemize}
A non-leaf vertex $v$ of the tree has two children $p_1$ and $p_2$, such that
\begin{itemize}
\item
$N(p_1) = {v\in N(p)\cap ext(\tilde{S}(p))}$, and
\item
$N(p_2) = {v\in N(p)\cap int(\tilde{S}(p))}$,
\end{itemize}
where $\tilde{S}$ denotes the cycle corresponding to a separator $S$.
For a leaf node $p$ of RGD, $N(p)$ contains only one node of $G$. In practice,
$N(p)$ may contain a small number of nodes, such that the distances in the
subgraph induced by $N(p)$ for every pair of nodes in $N(p)$ are pre-calculated and
stored in a table support $O(1)$ time look-up.

\noindent\paragraph{Range Minimum Query.} The range minimum query problem is to preprocess an array of length $n$
in $O(n)$ time such that all subsequent queries asking for the position of a minimal element between two specified
indices can be answered quickly. This can be done in constant time using no more than $2n+o(n)$ bits
\cite{Fischer2007}.

\subsection{Notation}

\noindent\paragraph{Projection.} Given a set of nodes $S$, and a node $v$, we define the projection
of $v$ on $S$ as the node in $S$ that is closest to $v$.

\noindent\paragraph{Radius} A graph has radius $\rho$ \emph{iff.} it has a shortest path tree with at most
$r$ levels.

\section{($1+\epsilon$)-Stretch, $O(\frac{1}{\epsilon}\log n\log \rho)$ Query Time Oracle}
\subsection{Preprocessing}
Find the node node $r\in G$ whose shortest path tree has $\rho$ levels, compute
the shortest-path tree $T$ in $G$ rooted at $r$ and based on $T$.
Calculate the RGD. Then \textbf{store},
\begin{enumerate}
\item
a table records, for each node $v\in G$, the leaf node $p\in$RGD, s.t. $v \in N(p)$;
\item
a table records, for each node $p\in$RGD, the depth of $p$ in RGD;
\item
a representation of RGD support quick ($O(1)$ time) computation of lowest common ancestor (lca);
\item
a table $T_v$ for each node $v\in G$ records, for each $p\in$RGD such that $v\in N(p)$, two sub-tables
for the paths $P'$, $P''$ in the separator $S(p)$, respectively. In details, $T_v[p][P']$ (similar for $T_v[p][P'']$)
consists of a sequence of $O(\frac{1}{\epsilon})$ pairs $(d_{-q}, h_{-q}), \ldots, (d_{0}, h_{0}), \ldots, (d_{w}, h_{w})$,
where $d_i$ is the distance from $v$ to a node $z_i$ on $P'$ and $h_i$ is the distance from $z_i$ to $r$ (root of the shortest path tree $T$),
such that the sequence has the \emph{distance property}: for any node $w$ on $P'$. there is a node $z_i$ such that
the distance from $v$ to $z_i$ plus the distance from $z_i$ to $w$ is at most $(1+\epsilon)$ times the distance from $v$ to $w$.
Refer to the nodes $z_i$ as \emph{portals}, and to the corresponding $d_i$ as \emph{portal distances}. In \cite{Klein2002},
it is proved that $O(\frac{1}{\epsilon})$ portals is enough to promise the distance property.
To be self-contained, we include in the appendix a simple version of the proof (See Appendix~\ref{app1}).
\end{enumerate}
Parts $1$ to $3$ need $O(n)$ \textbf{space}. Part $4$ needs $O(\frac{n}{\epsilon}\log n)$ \textbf{space}.
In addition, we \textbf{store} the portals for each label $\lambda$. In details, we store
\begin{enumerate}
\item[5.]
a table $T_{\lambda}$ for each label $\lambda \in L$, in which there is an entry for each piece $p\in$RGD, s.t.
$N(p)\cap V(\lambda)\neq \emptyset$. In each entry, it stores two sequence for the paths $P'$ and $P''$ forming $S(p)$,
respectively. In details, the sequence for $P'$ stores all the portals on $P'$ for nodes in $N(v)\cap V(\lambda)$, in
the increasing order according to their distances from the root $r$.

\item[6.] a hash table for each label $\lambda$ indicates whether there is an entry for a given separator in $T_{\lambda}$ and
return the index in $T_{\lambda}$ if yes (both operation could be done in constant time).
\end{enumerate}
Part $5$ needs $O(n\log n)$ \textbf{space} in total. Part $6$ needs $O(n\log n)$ \textbf{space} in total, since $T_{\lambda}$ contains one separator $S(p)$ \emph{iff.} there exists
at least one node with label $\lambda$ in $N(p)$.

\subsection{Query}
Given a node $u\in G$ and a label $\lambda\in L$, do as Algorithm~\ref{algo1}.

\SetAlgorithmName{Algorithm}{Name}{Algorithm}
\begin{algorithm}[H]
\SetAlgoLined
    \KwIn{$u$, $\lambda$}
\quad\\
\textbf{Initialization: $d(u,\lambda) \leftarrow \infty$}\\
\For{Each $p\in$RGD s.t. $u\in N(p)$ and $T_{\lambda}$ has an entry for $p$}{
	\For{Each path $P$ of $S(p)$}{
		\For{Each path portal $z_u$ of $u$ on $P$}{
			$C^+ \leftarrow $ $\{\lambda$'s portals on $P$ that is farther or equal than $z_u$ from $r\}$\\
			$\{z^+, v^+\} \leftarrow $ the portal of some $\lambda$ labeled node $v$ that
			achieves $\min \{\delta(v,z_v) + h(z_v)\}$ over $C^+$, and $v$\\
			$C^- \leftarrow $ $\{\lambda$'s portals on $P$ that is closer or equal than $z_u$ from $r\}$\\
			$\{z^-, v^-\} \leftarrow $ the portal of some $\lambda$ labeled node $v$ that
			achieves $\min \{\delta(v,z_v) - h(z_v)\}$ over $C^-$, and $v$\\
			$d' \leftarrow \{\delta(u,z_u) + \delta(z_u,z^+) + \delta(v,z^+),
			\delta(u,z_u) + \delta(z_u,z^-) + \delta(v,z^-)\}$\\
			$d(u,\lambda) \leftarrow \min\{d', d(u,\lambda)\}$\\
		}
	}
}

\KwOut{$d(u,\lambda)$}
\caption{}
\label{algo1}
\end{algorithm}

\begin{lemma}
Given $u$, $\lambda$, let $v$ be the $\lambda$ labeled node satisfying $\delta(u,v) = \delta(u,\lambda)$.
There exist a portal $z_u$ of $u$ and a portal $z_{v}$ of $v$ on the same path $P$, such that
$\delta(u, z_u)+\delta(z_u, z_v)+\delta(z_v, v)\leq (1+\epsilon)\delta(u,\lambda)$.
\end{lemma}
\begin{proof}
Let $p_u$, $p_{v}$ be the lowest pieces in RGD containing $u$, $v$, respectively,
i.e. $u\in N(p_u)$ and $v\in N(p_{v})$. Let $p_{uv}$ be the lca of $p_u$ and $p_v$ in RGD.
Then $u\in N(p_{uv})$, $v\in N(p_{uv})$, and the shortest
path from $u$ to $v$ crosses with $S(p_{uv})$. Denote the crossing point as $c$.
There exists a $u$'s portal $z_u$, such that $\delta(u,z_u)+\delta(z_u,c)\leq (1+\epsilon)\delta(u,c)$,
and a $v$'s portal $z_v$, such that $\delta(v,z_v)+\delta(z_v,c)\leq (1+\epsilon)\delta(v,c)$.

Hence $\delta(u, z_u)+\delta(z_u, z_v)+\delta(z_v, v) \leq (1+\epsilon)\delta(u,\lambda)$.
\end{proof}
This lemma implies that the output of Algorithm~\ref{algo1} achieves the $(1+\epsilon)$-approximation to
$\delta(u,\lambda)$, since
\begin{itemize}
\item
if $z_v$ is farther than $z_u$ from $r$,
then $h(z_v)+\delta(z_v,v)\geq h(z^+) + \delta(z^+, v^+))$, and hence
\begin{eqnarray*}
&&\delta(u,z_u) + \delta(z_u,z^+) + \delta(v^+,z^+) \\
&\leq& \delta(u,z_u) + \delta(z_u,z_v) + \delta(v,z_v)\\
&\leq& (1+\epsilon)\delta(u,\lambda);
\end{eqnarray*}
\item
if $z_v$ is closer than $z_u$ from $r$,
then $-h(z_v)+\delta(z_v,v)\geq -h(z^-) + \delta(z^-, v^-))$, and hence
\begin{eqnarray*}
&&\delta(u,z_u) + \delta(z_u,z^-) + \delta(v^-,z^-) \\
&\leq& \delta(u,z_u) + \delta(z_u,z_v) + \delta(v,z_v)\\
&\leq& (1+\epsilon)\delta(u,\lambda).
\end{eqnarray*}
\end{itemize}

To show the query time $O(\frac{1}{\epsilon}\log n\log \rho)$, we only need to show that $v^+$ ($v^-$) and $z^+$
($z^-$) could be found in $O(\log \rho)$ time. Actually, this could be done by identifying the range of $C^+$ ($C^-$)
of the portals of $\lambda$ on the specified path, using $O(\log \rho)$ time, and locating $v^+$ ($v^-$) by
range minimum query, using $O(1)$ time.

\begin{theorem}
Given an undirected planar graph $G=(V,E)$ and a label set $L$, each vertex $v\in V$ is attached with
one label in $L$. For any $0< \epsilon < 1$, there exists an oracle that could answer
any vertex-label query with stretch $1+\epsilon$, in $O(\frac{1}{\epsilon}\log n\log \rho)$ time.
The oracle needs $O(\frac{1}{\epsilon}n\log n)$ space.
\end{theorem}

\subsection{$3$-Stretch, $O(\log n\log \rho)$-Query Time Oracle}

Consider the case that $\epsilon = 2$. The oracle supports the $3$-stretch, $O(\log n\lg \rho)$-query time,
using space $O(n\log n)$. The \emph{space, query time product} (suggested by Christian Sommer \cite{Sommer2012})
is $O(n\log^2 n\log \rho)$, which is better than $O(n^{\frac{3}{2}})\times O(1)$ for general graphs.

Note that in this case, each node $u$ has only one portal on a specified path of a separator, which is the projection, denoted
by $z_u$, i.e. the node on the path closest to $u$. The reason is for any node $z$ on the path, we have
$\delta(z,z_u) \leq \delta(u,z_u)+ \delta(u,z) \leq 2\delta(u,z_u)$, and hence
$\delta(u,z) \leq \delta(u,z_u)+ \delta(z_u,z) \leq 3\delta(u,z_u)$.

\section{$O(1)$ Time to Identify $C^+$ ($C^-$) when $\rho = O(\log n)$}

In the case that $\rho = O(\log n)$, the time to identify $C^+$ ($C^-$) is $O(\log\log n)$. We show that
this could be improved to $O(1)$.

At first, note that when we store the portals for a label $\lambda$, it
is possible that a node servers as the portals for different nodes. It is obvious that we can only store
the one with the minimum portal-node distance. Thus fixed a label $\lambda$, on a path of a separator,
each node serves as at most one portal of $\lambda$. Using a word of $\rho = O(\log n)$ bits, denoted by $\omega$,
it can be identified whether a node on the path is a portal, i.e. the $i$-th bit is $1$ \emph{iff.}
the $i$-th node on the path is a portal for $\lambda$.
If the portals on a path for $\lambda$ are stored in the increasing order
of their positions on the path, its index could be retrieved by counting how many $1$ there are before the $i$-th
position of $\omega$. Since any operation on a single word is assumed to cost $O(1)$ time, we
achieve the $O(1)$ time method to identify $C^+$, with
\begin{itemize}
\item
$O(n\log n)$ space to record the position on the path forming separator, for each portal; and
\item
$O(n\log n)$ space to store $\omega$'s for all labels.
\end{itemize}

\begin{theorem}
Given an undirected planar graph $G=(V,E)$ and a label set $L$, each vertex $v\in V$ is attached with
one label in $L$. If the radius of $G$ is of order $O(\log n)$, then
for any $0< \epsilon < 1$, there exists an oracle that could answer
any vertex-label query with stretch $1+\epsilon$, in $O(\frac{1}{\epsilon}\log n)$ time.
The oracle needs $O(\frac{1}{\epsilon}n\log n)$ space.
\end{theorem}

\section{Label Changes}

We consider the cost to update the oracle, if a node $v$ changes its label from $\lambda_1$ to $\lambda_2$.
The portals of $v$ are not affected. However, the portals of $\lambda_1$ and $\lambda_2$ should be change.

To remove the portals of $v$ from the portals of $\lambda_1$, it requires to change the hash table indicating
of whether a separator is related to $\lambda_1$ for at most once, and change the portal sequences of $\lambda_1$
for at most $O(\log n)$ separators.

To add the portals of $v$ to the portals of $\lambda_2$, it requires to change the hash table indicating
of whether a separator is related to $\lambda_2$ for at most once, and change the portal sequences of $\lambda_2$
for at most $O(\log n)$ separators.

\section{Application}

\noindent\paragraph{Nearest Neighbor Search for Multiple Sets.}
Given a set $V(\lambda)$ of nodes, it is trivial to construct a linear size ($O(n)$) oracle
to support the query the nearest neighbor in $V(\lambda)$ for a query node $u$, i.e. the closest node to $u$ in $V(\lambda)$.
However, if there are several such sets $\{V(\lambda_i)\}_{\lambda_i\in L}$, this trivial method needs
$O(|L|\cdot n)$ space, which may be as big as $O(n^2)$ even each node is associated with only one label.
Using the oracle introduced in this document, we may construct an oracle using $O(\frac{1}{\epsilon}n\log n)$
space to support the query of $(1+\epsilon)$-NNS between a node and a label in $O(\frac{1}{\epsilon}\log n\log \rho)$
time.

Let's consider the case that each node in the graph could be associated with more than one labels.
In this case, $K = \sum_{\lambda_i \in L} |V(\lambda_i)|$ could be bigger than $n$, the oracle introduced
here needs $O(\frac{1}{\epsilon}K\log n)$ space.

Note that $O(|L|n) = \sum_{\lambda_i\in L} O(n)$ and
$O(\frac{1}{\epsilon}K\log n) = \sum_{\lambda_i\in L} O(\frac{1}{\epsilon}|V(\lambda_i)|\log n)$.
Hence the method introduced here is more efficient on space if $\frac{1}{\epsilon}|V(\lambda_i)| = o(\frac{n}{\log n})$
for all $\lambda_i\in L$.

\bibliographystyle{plain}
\bibliography{dist}

\appendix

\section{Finding Portals to Promise Distance Property}
\label{app1}

\begin{lemma}
For each node $v$ and each path $P'$, there exists a set $\{z_i\}$ of size less than $4(\epsilon - \epsilon^2)^{-1}$ for which the distance condition is satisfied.
\end{lemma}
\begin{proof}
Let $z_0$ be the node on $P'$ that is closest to the node $v$ and then we choose the remaining portals $z_i$ in two phases.
\begin{itemize}
\item
\textbf{Phase 1.} In this phase, we choose a set of nodes $z_i$ that are closer than $z_0$ to the root $r$, using Algorithm~\ref{algo2}.
Define a node in $z$ on $P'$ to be a \emph{candidate} with respect to (w.r.t.) $i$ iff.
\begin{enumerate}
\item
$z$ is closer tot eh root than $z_{i+1}$, and
\item
$\delta(v,z) < (1+\epsilon)^{-1}(\delta(v,z_{i+1} + \delta(r,z_{i+1}) + \delta(r,z))$.
\end{enumerate}

\SetAlgorithmName{Algorithm}{Name}{Algorithm}
\begin{algorithm}[H]
\SetAlgoLined

\textbf{Initialization: } $i \leftarrow -1$\\
\While{$\exists$ candidates w.r.t $i$}{
	$z_i \leftarrow$ candidate $z$ that is farthest from $r$\\
	$i \leftarrow i-1$
}
\caption{}
\label{algo2}
\end{algorithm}

Note the invariant for Phase 1: for $i < 0$ and any node $h$ lying strictly between
$z_i$ and $z_{i+1}$ on $P'$, we have $\delta(v,z_{i+1})+ \delta(z_{i+1}, h) \leq (1+\epsilon)\delta(v,h)$.
In particular, if there is no candidate w.r.t. $k$, then for any node $h$ lying strictly between
the root $r$ and $z_{i+1}$ on $P'$, we have $\delta(v,z_{i+1})+ \delta(z_{i+1}, h) \leq (1+\epsilon)\delta(v,h)$.
\item
\textbf{Phase 2.} In this phase, we choose a set of nodes $z_i$ that are farther than $z_0$ to the root $r$, using Algorithm~\ref{algo3}.
Define a node in $z$ on $P'$ to be a \emph{candidate} with respect to (w.r.t.) $i$ \emph{iff.}
\begin{enumerate}
\item
$z$ is farther to the root than $z_{i+1}$, and
\item
$\delta(v,z) < (1+\epsilon)^{-1}(\delta(v,z_{i-1} + \delta(r,z) + \delta(r,z_{i-1}))$.
\end{enumerate}

\SetAlgorithmName{Algorithm}{Name}{Algorithm}
\begin{algorithm}[H]
\SetAlgoLined

\textbf{Initialization: } $i \leftarrow 1$\\
\While{$\exists$ candidates w.r.t $i$}{
	$z_i \leftarrow$ candidate $z$ that is closest from $r$\\
	$i \leftarrow i+1$
}
\caption{}
\label{algo3}
\end{algorithm}

Note the invariant for Phase 2: for $i > 0$ and any node $h$ lying strictly between
$z_i$ and $z_{i-1}$ on $P'$, we have $\delta(v,z_{i-1})+ \delta(z_{i-1}, h) \leq (1+\epsilon)\delta(v,h)$.
In particular, if there is no candidate w.r.t. $k$, then for any node $h$ lying beyond
$z_{i-1}$ on $P'$, we have $\delta(v,z_{i-1})+ \delta(z_{i-1}, h) \leq (1+\epsilon)\delta(v,h)$.
\end{itemize}
Clearly, the $\{z_i\}$ chosen satisfies the distance condition. It remains to show that the number of $z_i$ chosen
is $O(\frac{1}{\epsilon})$. We show the analysis for $i>0$ and it applies to the case of $i<0$ in the similar way.

Since
\begin{eqnarray*}
\delta(v,z_i) &<& (1+\epsilon)^{-1}(\delta(v,z_{i-1}) + \delta(r,z_i) - \delta(r,z_{i-1}))\\
& \leq & (1+\epsilon)^{-1}\delta(v,z_{i-1}) + \delta(r,z_i) - \delta(r,z_{i-1})\\
& \leq & \delta(v,z_{i-1}) - (\epsilon - \epsilon^2)\delta(v,z_{i-1}) + \delta(r,z_i) - \delta(r,z_{i-1})\\
& \leq & \delta(v,z_{i-1}) - (\epsilon - \epsilon^2)\delta(v,z_{0}) + \delta(r,z_i) - \delta(r,z_{i-1}),
\end{eqnarray*}
then
\begin{eqnarray*}
\delta(v,z_i) - \delta(z_i,r) & < & \delta(v,z_{i-1}) - \delta(r,z_{i-1}) - (\epsilon - \epsilon^2)\delta(v,z_{0})\\
& < & \delta(v,z_{0}) - \delta(r,z_{0}) - i(\epsilon - \epsilon^2)\delta(v,z_{0})\\
\end{eqnarray*}
Noting $\delta(v,z_i) - \delta(z_i,r) \geq -\delta(v,z_{0}) - \delta(r,z_{0})$, it follows that $i < 2(\epsilon-\epsilon^2)^{-1}$.
This implies the lemma.
\end{proof}

\end{document}